\documentclass[conference]{IEEEtran}
\usepackage{graphicx}
\usepackage{float}
\usepackage{caption}
\usepackage{multicol}
\usepackage{mathrsfs}

\usepackage{amsmath}
\usepackage{amssymb}
\usepackage{amsthm}
\usepackage{multicol}
\usepackage{mathtools}
\usepackage{float}
\usepackage{placeins}
\usepackage{epstopdf}
\usepackage{multirow}
\usepackage{subfigure}
\usepackage{afterpage}
\usepackage{pdflscape}
\usepackage{float}
\restylefloat{table}

\usepackage{enumitem}
\usepackage[utf8]{inputenc}
\floatstyle{ruled}
\newfloat{algorithm}{tbp}{loa}
\providecommand{\algorithmname}{Algorithm}
\floatname{algorithm}{\protect\algorithmname}
\theoremstyle{remark}
\newtheorem{theorem}{Theorem}

\newtheorem{corollary}{Corollary}
\newtheorem{proposition}{Proposition}

\theoremstyle{remark}
\newtheorem{example}{Example}

\title{A Lifting Construction for Scalar Linear Index Codes}
\begin{document}
%\onecolumn
\author{Roop Kumar Bhattaram, Mahesh Babu Vaddi and B.~Sundar~Rajan, {\it Fellow,~IEEE}
}
\maketitle
%%%%%%%%%%%%%%%%%%%%%%%
\begin{abstract}
This paper deals with scalar linear index codes for canonical multiple unicast index coding problems where there is a source with  $K$ messages and there are $K$ receivers each wanting a unique message and having symmetric (with respect to the receiver index) antidotes (side information). Optimal scalar linear index codes for several such instances of this class of problems have been reported in \cite{MRRarXiv}. These codes can be viewed as special cases of the symmetric unicast index coding problems discussed in \cite{MCJ}. In this paper a lifting construction is given which constructs a sequence of multiple unicast index problems starting from a given multiple unicast index coding problem. Also, it is shown that if an optimal scalar linear index code is known for the problem given starting problem then optimal scalar linear index codes can be obtained from the known code for all the problems arising from the proposed lifting construction. For several of the  known classes of multiple unicast problems our construction is used to obtain several sequences of multiple unicast problem with optimal scalar linear index codes.\footnote{The authors are with the Department of Electrical Communication Engineering, Indian Institute of Science, Bangalore-560012, India. Email:bsrajan@ece.iisc.ernet.in.}
\end{abstract}
%%%%%%%%%%%%%%%%%
\section{Introduction and Background}
\label{sec1}

The problem of index coding with side information was introduced by Birk and Kol \cite{BiK}. Bar-Yossef \textit{et al.} \cite{YBJK} studied a type of index coding problem in which each receiver demands only one single message and the number of receivers equals number of messages. Ong and Ho \cite{OnH} classify the binary index coding problem depending on the demands and the side information possessed by the receivers. An index coding problem is unicast if the demand sets of the receivers are disjoint. If the problem is unicast and if the size of each demand set is one, then it is said to be single unicast. It was found that the length of the optimal linear index code is equal to the minrank of the side information graph of the index coding problem but finding the minrank is NP hard. \\

Maleki \textit{et al.} \cite{MCJ} found the capacity of symmetric multiple unicast index problem with neighboring antidotes (side information). In a symmetric multiple unicast index coding problem with equal number of $K$  messages and source-destination pairs, each destination has a total of $U+D=A<K$ antidotes, corresponding to the $U$ messages before (``up" from) and $D$ messages after (``down" from) its desired message. In this setting, the $k-$th receiver $R_{k}$ demands the message $x_{k}$ having the antidotes
%%%%%%%%
\begin{equation}
\label{antidote}
\{x_{k-U},\dots,x_{k-2},x_{k-1}\}~\cup~\{x_{k+1}, x_{k+2},\dots,x_{k+D}\}.
\end{equation}
%%%%%%%%%
The symmetric capacity of this index coding problem setting is shown to be as follows:
\begin{flushleft}
$U,D \in$ $\mathbb{Z},$\\
$0 \leq U \leq D$,\\
$U+D=A<K$,\ is\\
$C=\left\{
                \begin{array}{ll}
                  {1,\qquad\quad\ A=K-1}\\
                  {\frac{U+1}{K-A+2U}},A\leq K-2\qquad $per message.$
                  \end{array}
              \right.$
\end{flushleft}
\ \\
The above expression for capacity per message can be expressed as below for arbitrary $U$ and $D$:
%%%%%%%%%%%%%
\begin{equation}
\label{capacity}
C=\left\{
                \begin{array}{ll}
                  {1 ~~~~~~~~~~~~~~~~~~~~~~~~~~~~~ \mbox{if} ~~ U+D=K-1}\\
                  {\frac{min(U,D)+1}{K+min(U,D)-max(U,D)}} ~~~ \mbox{if} ~~U+D\leq K-2. 
                  \end{array}
              \right.
\end{equation}
%%%%%%%%%%%%%%%%%%%%%%%%%%%%%%%%

%The side information in the above index coding problem is represented by a directed graph G = ($V$,$E$) with $V = \{1,2,...,K\}$ is the set of vertices and E is the set of edges such that the directed edge $(i,j)\in E$ if receiver (destination) $R_{i}$ knows $x_{j}$. This graph $G$ for a given index coding problem is called side information graph. Let $G$ be a directed graph of $K$ vertices without self loops. A 0-1 matrix $A=(a_{i,j})$ fits in $G$ % for all $i$ and $j$ 
%if $a_{i,i}=1$ for all $i$ and $a_{i,j}$=0 whenever $(i,j)$ is not an edge of $G$. Let $rk_{2}()$ denotes the rank of the 0-1 matrix over $GF(2)$. The $minrank_{2}(G)$ is defined as \cite{YBJK}
%$$minrank_{2}(G) \triangleq min\{rk_{2}(A) : A \ fits \ in \ G\}.$$
%%%%%%%%%%%
%In a given index coding problem with side information graph $G$ = $(V,E)$, an edge $e\in E$ is said to be critical if the removal of $e$ from $G$ strictly reduce the capacity. The index coding problem $G=(V,E)$ is critical if every $e\in E$ is critical \cite{TSG}.
%%%%%%%%%%%%%%%%%%

In the setting of \cite{MCJ} with one sided antidote cases, i.e., the cases where $U$ or $D$ is zero,
without loss of generality, we can assume that $max(U,D)= D$ and $min(U,D)=0$ (all the results hold when $max(U,D)=U$), i.e.,
%%%
\begin{equation}
\label{antidotemcj}
{\cal K}_k =\{x_{k+1}, x_{k+2},\dots,x_{k+D}\}, 
\end{equation}
%%%
\noindent
for which \eqref{capacity} reduces to
%%%%%%%%
\begin{equation}
\label{capacity1}
C=\left\{
                \begin{array}{ll}
                  {1 ~~~~~~~~~~~~ \mbox{if} ~~ D=K-1}\\
                  {\frac{1}{K-D}} ~~~~~~ \mbox{if} ~~D\leq K-2. 
                  \end{array}
              \right.
\end{equation}
%%%%%%%%
symbols per message.
%%%%%%%%%%%%%%%%%%%%%%%%%%%%%%%%%%%%%%%%%%
\subsection{Contributions}
In \cite{MRRarXiv}  we considered ten cases of symmetric multiple unicast problems which are subclasses of the problem discussed in \cite{MCJ} with one sided antidotes. For each of these cases we proposed scalar linear codes that are optimal and also showed that these codes are capacity achieving. In this paper we start with a symmetric multicast index coding problem and construct a sequence of symmetric index coding problems indexed by $m$ where $m=2,3,\cdots $ The starting problem being for $m=1.$ We call this our {\it lifting construction}.  It is shown that if there is an optimal scalar linear index code for the starting problem then this code can be used to construct an optimal scalar linear index code for all the lifted problems, i.e., for all values of $m.$ We apply our lifting  construction to all the ten cases of the symmetric multiple unicast problems studied in \cite{MRRarXiv} and show several interesting outcomes like (i) One case gives another case when lifted with $m=2,$ and (ii) there are two cases such that lifting of problems  from these cases lead to problems of the same cases. In all other cases new classes of symmetric multicast problems are created for all of which an optimal scalar linear code is exhibited.  

%%%%%%%%%%%%%%%%%%%%%%%%%%%%%%%%%%
\section{Lifting Construction}
\label{sec2}

%%%%%%%%%%%%%%%%%%%%%%%%%%%%%%%%%%%%%%%
\begin{theorem}
\label{thm1}
For a multiple unicast index coding problem with $K$ messages $\{y_1,y_2,\cdots,y_K\}$ and the same number of receivers with the receiver $R_k$ wanting the message $y_k$  and having  a symmetric antidote pattern ${\cal K}_k$ given by 
%%%%%%%%%%%%%%%%%
\begin{equation}
\label{antidote1}
{\cal K}_k=\left\{
                \begin{array}{ll}
                  \ y_{k+a_1} \\
                  \ y_{k+a_2} \\
                  \ ~~~\vdots \\
                  \ y_{k+a_d} \\
                  \end{array}
              \right.
\end{equation}
%%%%%%%%%%
where $1 \leq k \leq K$ and  $a_1  < a_2 < \cdots < a_{d-1} < a_d < K,$ let $\mathfrak{C} = \{ t_1,t_2, \cdots, t_l \}$ be a scalar linear code of  length $l.$ \\
For an arbitrary positive integer $m,$  consider the index coding problem with  $mK$ number of messages $\{x_1,x_2,\cdots,x_{mK}\}$ and the number of receivers being $mK,$ and the receiver $R_k$ ($k=1,2, \cdots ,mK$) having antidote pattern given by 
\begin{equation}
\label{antidote2}
{\cal K}_k=\left\{
                \begin{array}{ll}
                  \ x_{k+K},x_{k+2K},\cdots,x_{k+(m-1)K}, \\
                  \ x_{k+a_1},x_{k+a_!+K},x_{k+a_!+2K},\cdots,x_{k+a_!+(m-1)K}\\
                  \ x_{k+a_2},x_{k+a_2+K},x_{k+a_2+2K},\cdots,x_{k+a_2+(m-1)K}\\
                  \ ~~~~~~~~~~~~~~~~~~~~~~~~ \vdots \\
                  \ x_{k+a_d},x_{k+a_d+K},x_{k+a_d+2K},\cdots,x_{k+a_d+(m-1)K}. \\
                  \end{array}
              \right.
\end{equation}
For this index coding problem the  code $\mathfrak{C}^{(m)}$  obtained by replacing every message symbol $y_k$ in the code symbols of $\mathfrak{C}$ with $\sum_{i=0}^{m-1} x_{k+iK} $ for $1 \leq k \leq K,$ i.e., by making the substitution $y_k= \sum_{i=0}^{m-1} x_{k+iK} $ is a valid scalar linear index code of length $l.$
\end{theorem}
%%%%%%%%%%%%%%%
%%%%%%%%%%%%%%%
\begin{proof}
In the set $\{x_{k}, x_{k+K}, x_{k+2K}, \dots ,x_{k+(m-1)K}\}$ (the set of message symbols appearing in  $x_{k}+x_{k+K}+x_{k+2K}+\dots+x_{k+(m-1)K}$), every message symbol $x_{k+jK},$ $0 \leq j \leq m-1,$  is in the antidote of every receiver $R_{k+iK},~~~ 0 \leq i \neq j \leq m-1,$  according to given antidote pattern in \eqref{antidote2}. In other words, 
\begin{equation}
\label{antidotestructure1}
{\cal K}_{k+jK} = ({\cal K}_k \cup x_k) \setminus \{x_{k+jK} \}. ~~~ 0 \leq j \leq m-1.
\end{equation}
\noindent
Denote by  $x_i^\prime$ the sum  $x_{i}+x_{i+K}+x_{i+2K}+\dots+x_{i+(m-1)K}.$ 
It is easily seen that 
\begin{equation}
\label{primestructure}
 x_{i}^\prime = x_{i+K}^\prime =  x_{i+2K}^\prime =  \dots  =    x_{i+(m-1)K}^\prime.
\end{equation}
\noindent
For every receiver whose wanted message symbol appears in the expression of $x_{k}^\prime ,$ i.e., appears in the set $\{x_{k}, x_{k+K}, x_{k+2K}, \dots ,x_{k+(m-1)K}\},$ every message symbol present in $x_{k+a_1}^\prime , x_{k+a_2}^\prime , \cdots , x_{k+a_d}^\prime$ are antidotes as per the given antidote pattern in \eqref{antidote2}.

Since $\mathfrak{C}^{(m)}$ is the index code obtained by making the substitution $y_k= \sum_{i=0}^{m-1} x_{k+iK}$ in the available code $\mathfrak{C},$ this substitution is essentially same as replacing $y_{i}$ with $x_{i}^\prime ~~ i=1,2, \cdots, K.$ 

First let us consider only the receivers $R_k, 1 \leq k \leq K.$ That the receivers $R_k, 1 \leq k \leq K$ can obtain  $y_k$ using the code $\mathfrak{C}$ means that the receivers can obtain the values of $x_k^\prime$ using the code $\mathfrak{C}^{(m)}.$  Then, since $R_k$ has all messages values in the expression of $x_k^\prime$ other than $x_k$ it can obtain its wanted message $x_k.$  

Next, for a fixed $k,~~ 1 \leq k \leq K$, all the receivers $R_{k+jK}, ~~ 1 \leq j \leq m-1,$ can decode their wanted messages due to the antidote structures given by \eqref{antidotestructure1} and \eqref{primestructure}.  Thus the code $\mathfrak{C}^{(m)}$ satisfy the demands of all the receivers for $k \in \{1,2,\dots,mK\}$. 

That the length of the code $\mathfrak{C}^{(m)}$ is same as that of the code $\mathfrak{C}$ is straight forward. This completes the proof.
\end{proof}

%Consider an index code by replacing $y_{i}$ with $x_{i}\prime=x_{i}+x_{i+K}+x_{i+2K}+\dots+x_{i+(m-1)K}$ for $i=1,2,\cdots,K$ in $\mathfrak{C_{1}}$. The code $\mathfrak{C_{1}}$ gives the guarantee for the decodability of all $K$ symbols of the form $y_{i}$ with the given antidote pattern (2). Thus the code $\mathfrak{C_{1}}$ gives the guarantee for the decodability of all $K$ symbols of the form $x_{i}\prime=x_{i}+x_{i+K}+x_{i+2K}+\dots+x_{i+(m-1)K}$ with the given antidote pattern in (3). We prove that every receiver $R_{k}$ can decode its wanted message $x_{k}$ from the above code for $k = \{1,2,\dots,Km\}$. Define $i$ = 1+($k-1$) mod $K$.
%To decode $x_{k}$, the $k^{th}$ receiver has to decode $x_{i}\prime=x_{i}+x_{i+K}+x_{i+2K}+\dots+x_{i+(m-1)K}$ because $x_{k}$ is present in $x_{i}\prime$. From $x_{i}\prime$,  receiver $R_{k}$ can decode $x_{k}$ because all other messages in $x_{i}+x_{i+K}+x_{i+2K}+\dots+x_{i+(m-1)K}$ are antidotes to $x_{k}$. Thus the code satisfy the demands of all the receivers for $k = \{1,2,\dots,mK\}$.
%\end{proof}
%%%%%%%%%%%%%%%%%%%%%%%%%%%%%%%%%%%%%%%%

The following Theorem \ref{thm2} identifies certain cases for which, in the lifting construction of Theorem \ref{thm1} if the code $\mathfrak{C}$ is of optimal length or equivalently a capacity achieving code then the lifted codes $\mathfrak{C}^{(m)}$ are also capacity achieving for all $m.$
%%%%%%%%%%%%%%%%%%%%%%%%%%%%%%%%%%%%%%%%%%%%%%%%%
%%%%%%%%%%%%%%%%%%%%%%%%%%%%%%%%%%%%%%%%%%%%%%%%%
\begin{theorem}
\label{thm2}
In the lifting construction of Theorem \ref{thm1} if the given code $\mathfrak{C}$ has length $l=K-a_d,$ then the given code and all its lifted codes are of optimal length and hence capacity achieving.

%Let $\mathfrak{C_{1}}$ is the scalar linear code achieving capacity $\frac{1}{K-D}$ with the number of messages $K$ ($y_{1},y_{2},\cdots,y_{K}$), number of receivers $K$, for the antidote pattern given in \eqref{antidote1} and maximum difference between the indices of required message and antidote of a receiver $R_{k}$ is $a_{d}=D$.
%The scalar linear code of optimal length with $Km$ number of receivers for some integer $m$, $Km$ number of messages $\{x_{1},x_{2},\cdots,x_{Km}\}$,  maximum difference between the indices of required message and antidote of a receiver $R_{k}$ is $Km-K+D$ and the antidote pattern as given in \eqref{antidote2} can be obtained by replacing every symbol $y_{k}$ in $\mathfrak{C_{1}}$ with \mbox{$x_{k}+x_{k+K}+x_{k+2K}+\dots+x_{k+(m-1)K}$}.
\end{theorem}
%%%%%%%%%%%%%%%
\begin{proof}
It is known \cite{MCJ} that the capacity of the multiple unicast index coding problem with one sided consecutive antidotes ${\cal K}_k = \{x_{k+1}, x_{k+2}, \cdots ,x_{k+D} \}$ is $\frac{1}{K-D}$ symbols per message or equivalently the optimal length of a index code is $K-D.$  Since the antidotes given by \eqref{antidote1} is a proper subset of the set of consecutive antidotes with $D=a_d$ the capacity of the given code in Theorem \ref{thm1} can be at most $\frac{1}{K-a_d}.$ Hence every linear index code with length $K-a_d$ is an optimal length code for the antidote pattern \eqref{antidote1}.

Now we proceed to establish  that the lifted codes are also of optimal length. It is easily seen from the antidote pattern given by \eqref{antidote2} that it is a proper subset of consecutive antidote pattern ${\cal K}_k = \{x_{k+1}, x_{k+2}, \cdots ,x_{k+(a_d+(m-1)K)} \}$ whose optimal length is $mK - (a_d+(m-1)K) = K-a_d.$ Since the lifting construction also gives the code of length $K-a_d$ for any $m,$ it follows that all the lifted codes are of optimal length.
%In Theorem 1, we proved that every receiver $R_{k}$ for $k=\{1,2,\cdots,Km\}$ could recover their wanted messages from the proposed construction. \\
%
%To prove the constructed code is optimal, the maximum difference between the indices of the required message $x_{k}$ of the receiver $R_{k}$ and its antidotes is $k+a_{d}+(m-1)K-k$=$a_{d}+(m-1)K$=$D+(m-1)K$. For the $k^{th}$ receiver ${\cal K}_k$ $\subseteq \{x_{k+1},x_{k+2},\cdots,x_{k+D+(m-1)K\}}$ for $k=\{1,2,\cdots,Km\}$. Capacity achieved by any code with  ${\cal K}_k$ $\subseteq \{x_{k+1},x_{k+2},\cdots,x_{k+D+(m-1)K\}}$ for $k=\{1,2,\cdots,Km\}$ should be less than or equal to  $\frac{1}{(Km)-(D+(m-1)K)}= \frac{1}{K-D}$. The code $\mathfrak{C_{1}}$ is achieving capacity $\frac{1}{K-D}$, thus $\mathfrak{C_{1}}$ uses $K-D$ code symbols. After lifting (replacing $y_{k}$ with $x_{k}'=x_{k}+x_{k+K}+\cdots+x_{k+(m-1)K}$), the number of code symbols is same as the number of code symbols in $\mathfrak{C_{1}}$. We use $K-D$ symbols to represent $Km$ message symbols. Thus the capacity achieved by the proposed method is $\frac{1}{K-D}$ and the code constructed is of optimal length.
\end{proof}
%%%%%%%%%%%%%
%%%%%%%%%%%%%%%%%%%%%%%%%%%%%%%%%%%%%%%%
When $d=1$ in the lifting construction of Theorem \ref{thm1}, it reduces to the following corollary.
%%%%%%%%%%%%%%%%%%%%%%%%%%%%%%%%%%%
\begin{corollary}
\label{cor1}
When $d=1$ in the lifting construction,  let $a_1=D.$ Then, for any $m$  the antidote \eqref{antidote2} for the lifted index coding problem  with $mk$ number of messages is
$${\cal K}_k=\left\{
                \begin{array}{ll}
                  \ x_{k+K},x_{k+2K},x_{k+3K},\cdots,x_{k+(m-1)K} \\
                                  \ x_{k+K+D},x_{k+2K+D},x_{k+3K+D},\cdots,x_{k+(m-1)K+D} \\ 
                  \end{array}
              \right.
$$
for $1 \leq k \leq mK.$ The lifted code is\\
$\mathfrak{C}^{(m)}=\{x_{i+(j-1)D}^\prime+x_{i+jD}^\prime|~ 1 \leq i \leq D, ~~ 1 \leq j \leq  \frac{K}{D}-1\}$
where \\
  $~~~~~~~x_{l}^\prime=x_{l}+x_{l+K}+x_{l+2K}+\dots+x_{l+(m-1)K}$ \\
for $l=1,2,\cdots,K$.
\end{corollary}
%%%%%%%%%%%%%%%%%%%%%%%%%%%%%%%%%%%%%%%
Notice that when $D$ in Corollary \ref{cor1} divides $K$ then the problem reduces to {\it Case I} of \cite{MRRarXiv} where it is shown that the code is of optimal length $K-D.$ Hence, from Theorem \ref{thm2} it follows that all the lifted codes are also of optimal length.

The following example illustrates the lifting construction of Theorem \ref{thm1} for $m=2$ and $m=3$ with the starting code being the case for which $d=1$ in \eqref{antidote1}.   
%%%%%%%%%
\begin{example}
\label{ex1}
Consider the case $K=20,~ d=1, a_1=4$ and \mbox{${\cal K}_k = \{x_{k+4}\}$  for $k=1,2, \cdots, 20.$} We have the code \\
 $\mathfrak{C}$=$\{x_1+x_5, ~ x_{2}+x_{6}, ~ x_{3}+x_{7}, ~ x_{4}+x_{8}, ~ x_5+x_9, \\
~~~~~ x_{6}+x_{10}, ~ x_{7}+x_{11}, ~ x_{8}+x_{12}, ~ x_9+x_{13}, ~ x_{10}+x_{14}, \\
~~~~~ x_{11}+x_{15}, ~ x_{12}+x_{16}, ~ x_{13}+x_{17}, ~ x_{14}+x_{18}, ~ x_{15}+x_{19},\\
~~~~~ x_{16}+x_{20}\}$.\\

\noindent
Let \mbox{$m=2$}. Then, for the lifted code  \mbox{$K=40$} and ${\cal K}_k = \{x_{k+20},~ x_{k+4},~ x_{k+24}\}$ for $k=1,2, \cdots, 40.$ The lifted code is given by \\
$\mathfrak{C}^{(2)}$=$\{x_{1}+x_{21}+x_{5}+x_{25},\ ~~x_{2}+x_{22}+x_{6}+x_{26},\\
~~~~~~~~ x_{3}+x_{23}+x_{7}+x_{27},\ ~~x_{4}+x_{24}+x_{8}+x_{28},\\
~~~~~~~~ x_{5}+x_{25}+x_{9}+x_{29},\ ~~x_{6}+x_{26}+x_{10}+x_{30},\\
~~~~~~~~ x_{7}+x_{27}+x_{11}+x_{31},\ ~x_{8}+x_{28}+x_{12}+x_{32},\\
~~~~~~~~ x_{9}+x_{29}+x_{13}+x_{33},\ ~x_{10}+x_{30}+x_{14}+x_{34},\\
~~~~~~~~ x_{11}+x_{31}+x_{15}+x_{35},\ x_{12}+x_{32}+x_{16}+x_{36},\\
~~~~~~~~ x_{13}+x_{33}+x_{17}+x_{37},\ x_{14}+x_{34}+x_{18}+x_{38},\\
~~~~~~~~ x_{15}+x_{35}+x_{19}+x_{39},\ x_{16}+x_{36}+x_{20}+x_{40}\}.$\\

\noindent
For $m=3,$ we have $K=60$ and\\
 $~~~~~{\cal K}_k = \{x_{k+20},~ x_{k+40},~ x_{k+4},~ x_{k+24},~ x_{k+44}\}$  \\ 
for $k=1,2, \cdots, 60.$  The lifted code is  \\
$\mathfrak{C}^{(3)}$=$\{x_{1}+x_{21}+x_{41}+x_{5}+x_{25}+x_{45},\\
~~~~~~~~ x_{2}+x_{22}+x_{42}+x_{6}+x_{26}+x_{46},\\
~~~~~~~~ x_{3}+x_{23}+x_{43}+x_{7}+x_{27}+x_{47},\\
~~~~~~~~ x_{4}+x_{24}+x_{44}+x_{8}+x_{28}+x_{48},\\
~~~~~~~~ x_{5}+x_{25}+x_{45}+x_{9}+x_{29}+x_{49},\\
~~~~~~~~ x_{6}+x_{26}+x_{46}+x_{10}+x_{30}+x_{50},\\
~~~~~~~~ x_{7}+x_{27}+x_{47}+x_{11}+x_{31}+x_{51},\\
~~~~~~~~ x_{8}+x_{28}+x_{48}+x_{12}+x_{32}+x_{52},\\
~~~~~~~~ x_{9}+x_{29}+x_{49}+x_{13}+x_{33}+x_{53},\\
~~~~~~~~ x_{10}+x_{30}+x_{50}+x_{14}+x_{34}+x_{54},\\
~~~~~~~~ x_{11}+x_{31}+x_{51}+x_{15}+x_{35}+x_{55},\\
~~~~~~~~ x_{12}+x_{32}+x_{52}+x_{16}+x_{36}+x_{56},\\
~~~~~~~~ x_{13}+x_{33}+x_{53}+x_{17}+x_{37}+x_{57},\\
~~~~~~~~ x_{14}+x_{34}+x_{54}+x_{18}+x_{38}+x_{58},\\
~~~~~~~~ x_{15}+x_{35}+x_{55}+x_{19}+x_{39}+x_{59},\\
~~~~~~~~ x_{16}+x_{36}+x_{56}+x_{20}+x_{40}+x_{60}\}.$\\
\end{example}
%%%%%%%%%%%%%

%All the ten cases ({\bf Case I} to {\bf Case X}) of codes listed in the previous section are of optimal length and from Theorem \ref{thm2} it follows that the lifted codes from any of these ten cases will also be of optimal length. 
%%%%%%%%%%%%%%%%%%%%%%%%%%%%%%%%%%%%%%%%%
\noindent
\section{Lifting construction for consecutive one sided antidote cases}
\label{sec3}

In this subsection Theorem \ref{thm1} is specialized for the two cases {\it Case VI} and {\it Case X}  of \cite{MRRarXiv} which have the one-sided consecutive antidote structure and stated in the form of Corollary \ref{cor2} and Corollary \ref{cor3} respectively. These are  special cases of the two sided consecutive  antidote structure studied \cite{MCJ}. 
%%%%%%%%%%%%%%%%%%%%%%%%%%%%%%%%%%%%%%%%%%%%%%%%%
\begin{corollary}
\label{cor2}

Consider the index coding problem with the antidote structure as in \eqref{antidote0} and  there is an integer $\lambda$ such that  $K-D$ divides $K-\lambda$ and  $\lambda$ divides $(K-D).$ 
\begin{equation}
\label{antidote0}
{\cal K}_k = \{x_{k+1}, x_{k+2}, \cdots, x_{k+D} \}
\end{equation}
In \cite{MRRarXiv} an optimal scalar linear code for this problem is shown to be the code \\
$\mathfrak{C}=\{x_{i}+x_{i+r}+\dots+x_{i+(q-1)r}+x_{qr+1+(i-1) mod \lambda) }\} \\ 
~~~~~~~~~~ |\ i = \{1,2,\dots,r\}\}$ \\
where $K-D=r$, and $\frac{K-\lambda}{K-D}=q$.

This case corresponds to Theorem \ref{thm1} with $d=D$ and $a_i=i$ for $1 \leq i \leq D.$   The lifting construction for this problem gives the code $\mathfrak{C}^{(m)}$ given by

%then  the proposed scalar linear code for the number of messages  $Km$ $\{x_{1},x_{2},\cdots,x_{Km}\}$ , number of receivers $Km$ and  $(Km-K+D)$ antidotes where the antidote pattern is as in \eqref{antidote1} is  
$\mathfrak{C}^{(m)}=\{x_{i}^\prime+x_{i+r}^\prime+\dots+x_{i+(q-1)r}^\prime+x_{qr+1+(i-1) mod \lambda }^\prime\} \\ 
~~~~~~~~~~ |\ i = 1,2,\dots,r\}$ \\
where \\
$ ~~~~~  x_{l}^\prime=x_{l}+x_{l+K}+x_{l+2K}+\dots+x_{l+(m-1)K}$\\
 for $l=1,2,\cdots,K$, $K-D=r$, and $\frac{K-\lambda}{K-D}=q$. \\
From Theorem \ref{thm2} it follows that all the lifted codes are of optimal length.
\end{corollary}
%%%%%%%%%%%%%%%%%%%%%%%%%%%%%%%%%%
The following examples illustrates Corollary \ref{cor2}.
%%%%%%%%%%%%%%%%%%%%%%%%%%%%%%%%%%%
\begin{example}
\label{ex2}
Consider $K=21,~ D=17,~ \lambda=1, m=1$ and ${\cal K}_k$ for $k=1,2,\cdots,21$ is as in \eqref{antidote0}. The optimal  index  code proposed for this code in \cite{MRRarXiv}is\\
$\mathfrak{C}=\{x_{1}+x_{5}+x_{9}+x_{13}+x_{17}+x_{21}, \\
~~~~~~~~x_{2}+x_{6}+x_{10}+x_{14}+x_{18}+x_{21}, \\
~~~~~~~~x_{3}+x_{7}+x_{11}+x_{15}+x_{19}+x_{21}, \\
~~~~~~~~x_{4}+x_{8}+x_{12}+x_{16}+x_{20}+x_{21}\}.$\\
With  \mbox{$m=2$} The lifted problem has  \mbox{$K=42,~ D=38$} and ${\cal K}_k$ for $k=1,2,\cdots,42$ is as in \eqref{antidote2}. The lifted code is given by\\
\begin{small}
{
$\mathfrak{C}^{(2)}=\{x_{1}+x_{22}+x_{5}+x_{26}+x_{9}+x_{30}+x_{13}+x_{34}+x_{17}+x_{38}+x_{21}+x_{42},\\
x_{2}+x_{23}+x_{6}+x_{27}+x_{10}+x_{31}+x_{14}+x_{35}+x_{18}+x_{39}+x_{21}+x_{42},\\
x_{3}+x_{24}+x_{7}+x_{28}+x_{11}+x_{32}+x_{15}+x_{36}+x_{19}+x_{40}+x_{21}+x_{42},\\
x_{4}+x_{25}+x_{8}+x_{29}+x_{12}+x_{33}+x_{16}+x_{37}+x_{20}+x_{41}+x_{21}+x_{42}\}.$ \\
}
\end{small}
\end{example}
%%%%%%%%%%%%%%%%%%%%%%%%%%%%%
\begin{corollary}
\label{cor3}
Consider an index coding problem with the consecutive adjacent antidote structure given by \eqref{antidote0} and there is an integer $\lambda$ such that  $K-D$ divides $K+\lambda$ and $\lambda$ divides $K-D.$ An optimal index code for this case is known to be \cite{MRRarXiv} \\
$\mathfrak{C}=\{x_{k}+x_{k+r}+x_{k+2r}+ \dots+x_{k+(q-1)r}+x_{k+(q-1)r+\lambda}\\
~~~~~~+x_{k+(q-1)r+2\lambda}+ \dots + x_{k+(q-1)r+(s-2)\lambda}| k = 1,2,\dots,\lambda\} \\ 
\cup \{x_{k}+x_{k+r}+x_{k+2r}\dots+x_{k+(q-2)r}+x_{k+(q-1)r-\lambda}|\\
~~~~~~~~~~~~~~~~~~~~~~~~~~~~~~~~~~~~~~~~~ k = \{\lambda+1,\lambda+2,\dots,p\}\} \\ 
\cup \{x_{k}+x_{k+r}+x_{k+2r}+ \dots+x_{k+(q-2)r}+x_{k+(q-2)r+\lambda}+x_{k+(q-2)r+2\lambda}+\dots + x_{k+(q-2)r+(s-1)\lambda}|\ k = p+1,p+2,\dots,r\}$\\
where $K-D=r$, $K-D-\lambda=p$, $\frac{K+\lambda}{K-D}=q$  and \mbox{$\frac{K-D}{\lambda}=s$}.\\
For this problem lifting construction gives an index coding problem with  $mK$  number of messages  $\{x_{1},x_{2},\cdots,x_{Km}\}$, and antidotes as in \eqref{antidote2}. The lifted code is given by \\
$\mathfrak{C}^{(m)}=\{x_{k}^\prime+x_{k+r}^\prime+x_{k+2r}^\prime+ \dots+x_{k+(q-1)r}^\prime \\
~~ +x_{k+(q-1)r+\lambda}^\prime+x_{k+(q-1)r+2\lambda}^\prime+ \dots + x_{k+(q-1)r+(s-2)\lambda}^\prime \\ 
~~~~~~~~ | k = 1,2,\dots,\lambda\} \\ 
 \cup \{x_{k}^\prime+x_{k+r}^\prime+x_{k+2r}^\prime\dots+x_{k+(q-2)r}^\prime+x_{k+(q-1)r-\lambda}^\prime \\
~~~~~~~~ |\ k = \lambda+1,\lambda+2,\dots,p\} \\ 
\cup \{x_{k}^\prime+x_{k+r}^\prime+x_{k+2r}^\prime+ \dots+x_{k+(q-2)r}^\prime+x_{k+(q-2)r+\lambda}^\prime+x_{k+(q-2)r+2\lambda}^\prime+\dots + x_{k+(q-2)r+(s-1)\lambda}^\prime \\
~~~~~~~~ |\ k = p+1,p+2,\dots,r\}$\\ where $x_{l}^\prime=x_{l}+x_{l+K}+x_{l+2K}+\dots+x_{l+(r-1)K}$ for $l=1,2,\cdots,K$, $K-D=r$, $K-D-\lambda=p$, $\frac{K+\lambda}{K-D}=q$  and \mbox{$\frac{K-D}{\lambda}=s$}.\\
\end{corollary}
%%%%%%%%%%%%%%%%%%%%%%%%%%%%%%%%%%%%
Corollary \ref{cor3} is illustrated in the following example.
%%%%%%%%%%%%%%%%%%%%%%%%%%%%%%%%%%%%
\begin{example}
\label{ex3}
To illustrate Corollary \eqref{cor3} consider $K=28,~ D=18,~ \lambda=2, m=1$ and ${\cal K}_k$ for $k=1,2,\cdots,28$ is as in \eqref{antidote0}. We get the code\\
$\mathfrak{C}=\{x_{1}+x_{11}+x_{21}+x_{23}+x_{25}+x_{27},\\ x_{2}+x_{12}+x_{22}+x_{24}+x_{26}+x_{28},\\ 
x_{3}+x_{13}+x_{21},~~ x_{4}+x_{14}+x_{22},~~ x_{5}+x_{15}+x_{23},\\ 
x_{6}+x_{16}+x_{24},~~ x_{7}+x_{17}+x_{25},~~ x_{8}+x_{18}+x_{26},\\ 
x_{9}+x_{19}+x_{21}+x_{23}+x_{25}+x_{27},\\ 
x_{10}+x_{20}+x_{22}+x_{24}+x_{26}+x_{28}\}.$\\

Lifting construction with $m=2$ leads to the index coding problem with  \mbox{$K=56$} and ${\cal K}_k$ for $k=1,2,\cdots,56$ is as in \eqref{antidote2}. The lifted code is given by\\
\begin{small}
{
$\mathfrak{C}^{(2)}=\{x_{1}+x_{29}+x_{11}+x_{39}+x_{21}+x_{49}+x_{23}+x_{51}+x_{25}+x_{53}+x_{27}+x_{55},\\
x_{2}+x_{30}+x_{12}+x_{40}+x_{22}+x_{50}+x_{24}+x_{52}+x_{26}+x_{54}+x_{28}+x_{56},\\
x_{3}+x_{31}+x_{13}+x_{41}+x_{21}+x_{49},\\
x_{4}+x_{32}+x_{14}+x_{42}+x_{22}+x_{50},\\
x_{5}+x_{33}+x_{15}+x_{43}+x_{23}+x_{51},\\
x_{6}+x_{34}+x_{16}+x_{44}+x_{24}+x_{52},\\
x_{7}+x_{35}+x_{17}+x_{45}+x_{25}+x_{53},\\
x_{8}+x_{36}+x_{18}+x_{46}+x_{26}+x_{54},\\
x_{9}+x_{37}+x_{19}+x_{47}+x_{21}+x_{49}+x_{23}+x_{51}+x_{25}+x_{53}+x_{27}+x_{55},\\
x_{10}+x_{38}+x_{20}+x_{48}+x_{22}+x_{50}+x_{24}+x_{52}+x_{26}+x_{54}+x_{28}+x_{56}.$\\
}
\end{small}
\end{example}
%%%%%%%%%%%%%%%%%%%%%%%%%%%%%%%%%%%%%%%%%%%%%%%%%%%%

%%%% ??????  %%%%%%%%%%%%%%%%%%%%%%%%%%%%%%%%%%%%%%%%%%%%%%%%%%%%%%%%%%%%%%%%%%%%%%%%%%%%%%%%%%%%%
\noindent
\section{Interrelationship among some classes of problems}
\label{sec4}

In this subsection we identify (i) a class of uniprior problems whose appropriate lifting gives another known class of known problems and (ii) two classes of problems which when lifted give the same classes of problems. 

%%%%%%%%%%%%%%%%%%%%%%%%%%%%%%%%%%%%%%%%%
\begin{proposition}
\label{pro1}
The index coding problem {\it Case (a):}  $D$ divides $K$ and ${\cal K}_k =\{x_{k+D}\}$ when lifted with $m=2$  leads to the class of problems defined as follows:
{\it Case (b):} $D-\frac{K}{2}$ divides $\frac{K}{2} $ and \\
 $${\cal K}_k=\{x_{k+\frac{K}{2}},x_{k+D-\frac{K}{2}},x_{k+D}\}.$$
\end{proposition}
\begin{proof}
It is known \cite{MRRarXiv} that a problems of {\it Case (a)} has the optimal index code \\
{\small
$$\mathfrak{C}_a=\{x_{i+(j-1)D}+x_{i+jD}|~i=1,2,\dots,D,~ j=1,2,\dots, \frac{K}{D}-1\}.$$
}
and a problem of {\it Case (b)} has the optimal length code \\ 
$\mathfrak{C}_b=\{x_{i+jr}+x_{\frac{K}{2}+i+jr}+x_{i+(j+1)r}+x_{\frac{K}{2}+i+(j+1)r} \\
~~~~~~~~~~  ~| ~i =\{1,2,\dots,r\},\ j=\{0,1,2,\dots,n-2\}\}$ \\
where $D-\frac{K}{2}=r$ and $\frac{K/2}{D-\frac{K}{2}}=n$.

For the given code with $K$ messages symbols and ${\cal K}_k = x_{k+D},$ when lifted with $m=2$ results in the number of messages being $2K$ and ${\cal K}_k = \{x_{k+K},~ x_{k+D},~ x_{k+D+K}\}.$  Clearly the resulting code is of type {\it Case (b)}. The number of antidotes are $2K-(K-D)=K+D$. The code in {\it Case (a)} is given by $\mathfrak{C}=\{y_{i+(j-1)D}+y_{i+jD}|~i=1,2,\dots,D,~ j=1,2,\dots, \frac{K}{D}-1\}.$ \\
If we replace every symbol $y_{i}$ in the above code with $x_{i}+x_{i+K}$, we get the code for the {\it Case (b)}.
\end{proof}
%%%%%%%%%%   The following proposition involves dual code !!! %%%%%%%%%
%\begin{proposition}
%\label{lem2}
%Let $\mathfrak{C_{3}}$ is the code for given $K$ and $D$ satisfying the conditions in $case III$ and let $\mathfrak{C_{4}}$ is the code for given $K'=K$ and $D'=K-D$ satisfying the conditions in $case IV$. $\mathfrak{C_{3}}$ and $\mathfrak{C_{4}}$ are dual to each other if the message symbols are from the field of characteristic two.
%\end{proposition}
%%
%\begin{proof}
%The $L$ matrix for $(\mathfrak{C_{3}})$ is a four way diagonal matrix with the order $K \times K-D$ and with four ones in each column appearing in $i^{th}, (i+D)^{th}, (i+\frac{K}{2})^{th}$ and $(i+\frac{K}{2}+D)^{th}$ positions in a cyclic manner.  The $L$ matrix for $(\mathfrak{C_{4}})$ is a $\frac{D}{\frac{K}{2}-D}$ diagonal matrix with the order $K' \times K'-D'=K \times D$ and with two ones in each column are separated by $\frac{K}{2}-D$ positions.
%Every column  in the $L$ matrix of $\mathfrak{C_{3}}$ overlaps in either zero positions or two positions with every column of $L$ matrix of $\mathfrak{C_{4}}$. Thus every column of $L$ matrix of $\mathfrak{C_{3}}$ is orthogonal to every column of $L$ matrix of $\mathfrak{C_{4}}$ over the field of characteristic two. The rank of $L$ matrix for $\mathfrak{C_{3}}$ is $K-D$ and the rank of the $L$ matrix for $\mathfrak{C_{4}}$ is $D$. The sum of rank of both the matrices is $K$ implies that one code is the dual of other code.
%\end{proof}
%%%%%%%%%%%%
From Corollary \ref{cor1} and Example \ref{ex1} we see that lifting construction gives new index coding problems which do not belong to the class of problems which is lifted. However, for the problems of {\bf Case II} and {\bf Case VIII} in \cite{MRRarXiv} the lifting construction leads to problems of the same class. This is shown in the following two propositions.

%%%%%%%%%%%%%%%%%%%%%%%%%%%
\begin{proposition}
\label{pro2}
Consider the index coding problem:  \mbox{$K-D$} divides $K$ and the antidote pattern is given by \eqref{antidote2a}.
\begin{equation}
\label{antidote2a}
{\cal K}_k=\{x_{k+K-D},x_{k+2(K-D)},...,x_{k+D}\}
\end{equation}
An optimal length index code for this problem is \cite{MRRarXiv}
 $${\mathfrak{C}}=\{x_{i}+x_{i+r}+\dots+x_{i+(n-1)r}\ |\ {i = \{1,2, \dots, r\}\}}$$ where  $K-D=r$ and \mbox{$\frac{K}{K-D}=n$}. 
For this class of problems lifting construction gives raise to  index coding problems of the same case. In other words, a  problem of this class with parameters $K$ and $D$ under lifting construction for any $m$ leads to problems with parameters  $K^\prime$ and $D^\prime$, which falls in the same case and same antidote structure. 
%That is, lifting of the codes in $case II$ and $case VIII$ can not exploit any values of $K'$ and $D'$ which does not fall in any of the given ten cases.
\end{proposition}
\begin{proof}
For the problem to be lifted the  number of messages is $K$ and the maximum difference between the indices of required message and antidote for each of the receiver is $a_d=D.$ The condition that the parameters $K$ and $D$ satisfy is that $K-D$ divides $K$. After lifting, number of messages is $K^\prime = mK$ and the maximum difference between the indices of required message and antidote for each of the receiver is $D^\prime=(m-1)K+D.$ We have $K^\prime-D^\prime=K-D$. Because $K-D$ divides $K$, $K^\prime-D^\prime=K-D$ divides $mK$ and thus the values of $K^\prime$ and $D^\prime$ satisfy the requirement of the class under consideration.  The antidote pattern after lifting is
\begin{equation}
\small
{\cal K}_k=
                \begin{array}{ll}
                  \ \{x_{k+K},x_{k+2K},\cdots,x_{k+(m-1)K}, \\
                  \ x_{k+(K-D)},x_{k+(K-D)+K},\cdots,x_{k+(K-D)+(m-1)K},\\
                  \ x_{k+2(K-D)},x_{k+2(K-D)+K},\cdots,x_{k+2(K-D)+(m-1)K},\\
                  \ ~~~~~~~~~~~~~~~~~~~~~~~~ \vdots \\
                  \ x_{k+p(K-D)},x_{k+p(K-D)+K},\cdots,x_{k+p(K-D)+(m-1)K} \} \\
                  \ ~~~~~~~~~~~~~~~~~~~~~~~~~~~~~~~~where \ p=\frac{D}{K-D}
                  \end{array} 
\label{adcase2lifting}           
\end{equation}
for $1 \leq k \leq mK.$  The antidote pattern for the problem of the given class  with parameters $mK$ and $(m-a)K+D$  is
\begin{equation}
{\cal K}_k=\{x_{k+K-D},x_{k+2(K-D)},x_{k+2(K-D)},...,x_{k+(m-1)K+D}\}.
\label{adcase2nolifting} 
\end{equation}
The antidote patterns mentioned in \eqref{adcase2lifting} and \eqref{adcase2nolifting} are exactly same. 
This completes the proof. 
\end{proof}
%%%%%%%%%%%%%%%%%%%%%%%%%%%%%%%%
\begin{proposition}
\label{pro3}
Consider the class of index coding problems: There is an integer $\lambda$ such that $K-D+\lambda$ divides $K$ and $\lambda$ divides $k-D$ and
\begin{equation}
\label{antidote8}
{\cal K}_k=\left\{
                \begin{array}{ll}
                  \{x_{k+\lambda},\ x_{k+\lambda+(K-D)},\\x_{k+2\lambda+(K-D)},\ x_{k+2\lambda+2(K-D)},\\ ~~~~~~~~~\vdots ~~~~~~~~~~~~~~~~~~~ \vdots\\x_{k+(p-1)\lambda+(p-2)(K-D)},\ x_{k+(p-1)\lambda+(p-1)(K-D)},\\x_{k+p\lambda+(p-1)(K-D)}\}\\
                  \end{array}
              \right.
\end{equation}
where $p=\frac{K}{K-D+\lambda}$.\\
For a index coding problem of this class lifting construction creates index coding problems of the same class.
% In other words, for a problem of {\bf Case VIII} with parameters $K$ and $D$ the lifting construction with any value for  $m$ leads to a problem with parameters  $K^\prime$ and $D^\prime$, which falls in the same case  and same antidote structure.
\end{proposition}
%%    ^\prime   
\begin{proof}
Before lifting, the condition required to satisfy in the class is $K-D+\lambda$ divides $K$ and $\lambda$ divides $K-D$. After lifting, we have $K^\prime =mK$, $D^\prime=(m-1)K+D$ and $K^\prime-D^\prime=K-D$. Because  $K-D+\lambda$ divides $K$, $K^\prime-D^\prime+\lambda=K-D+\lambda$ divides $K$ and thus the values of $K^\prime$ and $D^\prime$ satisfy the requirement to fall in the given class. 
The antidote pattern required after lifting is
\begin{equation}
\small
{\cal K}_k=
                \begin{array}{ll}
                  \ \{x_{k+K},x_{k+2K},\cdots,x_{k+(m-1)K}, \\
                  \ x_{k+\lambda},x_{k+\lambda+K},\cdots,x_{k+\lambda+(m-1)K}, \\
                  \ x_{k+\lambda+(K-D)},x_{k+\lambda+(K-D)+K},\cdots,x_{k+\lambda+(K-D)+(m-1)K},\\
                  \ x_{k+2\lambda+(K-D)},x_{k+2\lambda+(K-D)+K},\cdots,x_{k+2\lambda+(K-D)+(m-1)K},\\
                  \ ~~~~~~~~~~~~~~~~~~~~~~~~ \vdots \\
                  \ x_{k+D},x_{k+D+K},\cdots,x_{k+D+(m-1)K} \} \\
                  \end{array}  
\label{adcase8lifting}           
\end{equation}

%The antidote pattern by considering the lifted index coding problem is
Repeating \eqref{antidote8} with $K$ replaced by $mK$ and $D$ replaced by $(m-1)K+D$ gives
\begin{equation}
{\cal K}_k=\{x_{k+\lambda},x_{k+\lambda+K-D},x_{k+2\lambda+K-D},...,x_{k+(m-1)K+D}\}
\label{adcase8nolifting} 
\end{equation}
which is same as \eqref{adcase8lifting}, i.e., the antidote patterns mentioned in \eqref{adcase8lifting} and \eqref{adcase8nolifting} are exactly same. This completes the proof.
\end{proof}
%%%%%%%%%%%%%%%%%%%%%%%%%%%%%%%%%%%%
\section{Lifting construction for four classes}
\label{sec5}

In this section we demonstrate the lifting construction for four more classes of index coding problems apart from the classes of problems discussed in the previous sections. These four classes have known optimal index codes \cite{MRRarXiv} and generate new classes of index coding problems under lifting which have explicit optimal linear index codes. \\

\noindent
{\bf Class (i):} Consider the index coding problems for which  $\frac{K}{2}-D$ divides $D$ with the antidote pattern given in \eqref{antidote4}, 
\begin{equation}
\label{antidote4}
{\cal K}_k=\{x_{k+\frac{K}{2}-D},x_{k+2(\frac{K}{2}-D)}+ \cdots +x_{k+D}\}.
\end{equation}
the optimal proposed scalar linear code is \cite{MRRarXiv},\\
$\mathfrak{C}=\{x_{i}+x_{i+r}+\dots+x_{i+pr}, \\ 
~~~~~~~~ x_{i+r}+x_{i+2r}+\dots+x_{i+(p+1)r}, \\
~~~~~~~~~~~~~~~~~~~~ \vdots \\
~~~~~~~~ x_{i+r(p+1)}+x_{i+r(p+2)}+\dots+x_{i+(n-1)r} \\
~~~~~~~~~~~~~~~~~~~~~~~ | i=\{1,2, \dots ,r\}\}$

\noindent
where  $r=\frac{K}{2}-D$, $n=\frac{K}{\frac{K}{2}-D}$ and $\frac{D}{\frac{K}{2}-D}=p$. \\
\begin{corollary}
\label{cor4}
For this class of codes lifting construction leads to the class with $mK$ messages and the antidote structure 
$${\cal K}_k=\left\{
                \begin{array}{ll}
                  \ x_{k+K},x_{k+2K},x_{k+3K},\cdots,x_{k+(m-1)K} \\
                  \ x_{k+r},x_{k+r+K},x_{k+r+2K},\cdots,x_{k+r+(m-1)K}\\
                  \ x_{k+2r},x_{k+2r+K},x_{k+2r+2K},\cdots,x_{k+2r+(m-1)K} \\
                  \ ~~~~~~~~~~~~~~~~~~~~~ \vdots \\
                  \ x_{k+pr},x_{k+pr+K},x_{k+pr+2K},\cdots,x_{k+pr+(m-1)K}. \\
                  \end{array}
              \right\}
$$
 The optimal length scalar lifted index  code is given by,\\
$\mathfrak{C}=\{x_{i}^\prime +x_{i+r}^\prime +\dots+x_{i+pr}^\prime , \\ 
~~~~~~~~ x_{i+r}^\prime +x_{i+2r}^\prime +\dots+x_{i+(p+1)r}^\prime , \\
~~~~~~~~~~~~~~~~~~~~ \vdots \\
x_{i+r(p+1)}^\prime +x_{i+r(p+2)}^\prime +\dots+x_{i+r(n-1)}^\prime  
 | i=1,2, \dots ,r\}$
where  $r=\frac{K}{2}-D$, $n=\frac{K}{\frac{K}{2}-D}$, $\frac{D}{\frac{K}{2}-D}=p$ and $x_{j}^\prime =x_{j}+x_{j+K}+x_{j+2K}+\dots+x_{j+(m-1)K}$ for $j=1,2,\cdots,K$. \\
\end{corollary}
%%%%-----------------------------------

\begin{example}
\label{ex4}
$K=20,\ D=8,\ m=1$ and $\mathcal{K}_{k}$ for $k=1,2,\dots,20$ is $\{x_{k+2}, x_{k+4},x_{k+6},x_{k+8}\}$. The proposed code is\\
$\mathfrak{C}=\{x_{1}+x_{3}+x_{5}+x_{7}+x_{9}, ~~ x_{2}+x_{4}+x_{6}+x_{8}+x_{10}, \\
x_{3}+x_{5}+x_{7}+x_{9}+x_{11}, ~~ x_{4}+x_{6}+x_{8}+x_{10}+x_{12},\\ 
x_{5}+x_{7}+x_{9}+x_{11}+x_{13}, ~~ x_{6}+x_{8}+x_{10}+x_{12}+x_{14},\\ 
x_{7}+x_{9}+x_{11}+x_{13}+x_{15}, ~~ x_{8}+x_{10}+x_{12}+x_{14}+x_{16},\\
x_{9}+x_{11}+x_{13}+x_{15}+x_{17}, ~~ x_{10}+x_{12}+x_{14}+x_{16}+x_{18},\\
x_{11}+x_{13}+x_{15}+x_{17}+x_{19}, ~~ x_{12}+x_{14}+x_{16}+x_{18}+x_{20}\}$.
\\

Let $m$=2. Then \mbox{$K=40,~ D=28$} and ${\cal K}_k$ for $k=1,2,\cdots,40$ is\\ \mbox{$\{x_{k+2}, x_{k+4},x_{k+6},x_{k+8},x_{k+20},x_{k+22}, x_{k+24},x_{k+26},x_{k+28}\}$}. The code is given by\\
$\mathfrak{C}=\{x_{1}+x_{21}+x_{3}+x_{23}+x_{5}+x_{25}+x_{7}+x_{27}+x_{9}+x_{29},\\
x_{2}+x_{22}+x_{4}+x_{24}+x_{6}+x_{26}+x_{8}+x_{28}+x_{10}+x_{30},\\
x_{3}+x_{23}+x_{5}+x_{25}+x_{7}+x_{27}+x_{9}+x_{29}+x_{11}+x_{31},\\
x_{4}+x_{24}+x_{6}+x_{26}+x_{8}+x_{28}+x_{10}+x_{30}+x_{12}+x_{32},\\
x_{5}+x_{25}+x_{7}+x_{27}+x_{9}+x_{29}+x_{11}+x_{31}+x_{13}+x_{33},\\
x_{6}+x_{26}+x_{8}+x_{28}+x_{10}+x_{30}+x_{12}+x_{32}+x_{14}+x_{34},\\
x_{7}+x_{27}+x_{9}+x_{29}+x_{11}+x_{31}+x_{13}+x_{33}+x_{15}+x_{35},\\
x_{8}+x_{28}+x_{10}+x_{30}+x_{12}+x_{32}+x_{14}+x_{34}+x_{16}+x_{36},\\
x_{9}+x_{29}+x_{11}+x_{31}+x_{13}+x_{33}+x_{15}+x_{35}+x_{17}+x_{37},\\
x_{10}+x_{30}+x_{12}+x_{32}+x_{14}+x_{34}+x_{16}+x_{36}+x_{18}+x_{38},\\
x_{11}+x_{31}+x_{13}+x_{33}+x_{15}+x_{35}+x_{17}+x_{37}+x_{19}+x_{39},\\
x_{12}+x_{32}+x_{14}+x_{34}+x_{16}+x_{36}+x_{18}+x_{38}+x_{20}+x_{40}\}$.\\
\end{example}

%%%%%%%%%%%%%%%%%%%%%%%%%%%%%%%%%%%%%%%%%%%%%
\noindent{
\bf Class (ii):} This class of index coding problems is defined as follows: There is an  integer $\lambda$ such that $D$ divides $K-\lambda$ and $\lambda$ divides $D$ and the antidote pattern is 
$${\cal K}_k=\left\{
                \begin{array}{ll}
                  \{x_{k+D}\},\ $if$\ k\leq K-D-\lambda\\
                  \{x_{k+\lambda},x_{k+2\lambda}\dots,x_{k+D}\},\ $if$\ K-D-\lambda<k\leq K
                  \end{array}
              \right.
$$
An optimal length index code for this class is known to be \cite{MRRarXiv} \\
$\mathfrak{C}=\{{x_{i+(j-1)D}+x_{i+jD}}|\ i = \{1,2,\dots,D\},\ j= \{1,2,\dots,n-1\}\}\\ \cup \{{x_{K-\lambda+r}+x_{K-\lambda+r-\lambda}+\dots+x_{K-\lambda+r-t\lambda}} \\
~~~~~~~~~~~|\ r = \{1,2,\dots,\lambda\}, \ t = \{1,2,\dots,\frac{D}{\lambda}$\}\} \\
 for $\frac{K-\lambda}{D}>1$  and $\frac{K-\lambda}{D}=n$.

%%%-----------------------------------------------
\begin{corollary}
\label{cor5}
The lifting construction for this class leads to an index coding problem with $mK$ messages having the side information pattern given below:\\ 
\noindent
For $1 \leq k \leq mK,$ if $(k \ mod \ K)\leq K-D-\lambda$ then 
$${\cal K}_k=\left\{
                \begin{array}{ll}
                  \ x_{k+K},x_{k+2K},x_{k+3K},\cdots,x_{k+(m-1)K} \\
                                  \ x_{k+K+D},x_{k+2K+D},x_{k+3K+D},\cdots,x_{k+(m-1)K+D}. \\ 
                  \end{array}
              \right.
$$
and if  $K-D-\lambda<(k \ mod \ K)\leq K$ then
$${\cal K}_k=\left\{
                \begin{array}{ll}
                  \ x_{k+K},x_{k+2K},x_{k+3K},\cdots,x_{k+(m-1)K} \\
                  \ x_{k+\lambda},x_{k+\lambda +K},x_{k+\lambda +2K},\cdots,x_{k+\lambda +(m-1)K}\\
                  \ x_{k+2\lambda},x_{k+2\lambda +K},x_{k+2\lambda +2K},\cdots,x_{k+2\lambda+(m-1)K} \\
                  \ \vdots \\
                  \ x_{k+D},x_{k+D+K},x_{k+D+2K},\cdots,x_{k+D+(m-1)K}. \\
                  \end{array}
              \right.
$$
The resulting lifted code $\mathfrak{C}^{(m)}$ is given by \\
$\mathfrak{C}^{(m)}=\{{x_{i+(j-1)D}^\prime +x_{i+jD}^\prime }|~ 1 \leq i \leq D,~ 1 \leq j \leq n-1 \} \\   
~~~~~ \cup \{{x_{K-\lambda+r}^\prime +x_{K-\lambda+r-\lambda}^\prime +\dots+x_{K-\lambda+r-t\lambda}^\prime } \\
~~~~~~~~~ |\ r = \{1,2,\dots,\lambda\},\ t = \{1,2,\dots,\frac{D}{\lambda}\}\} $\\
$~~~~~~~~~~~~~~~~~$ for $\frac{K-\lambda}{D}>1$, $\frac{K-\lambda}{D}=n$ and \\ 
$~~$ $x_{l}^\prime =x_{l}+x_{l+K}+x_{l+2K}+\dots+x_{l+(m-1)K}$  for $1 \leq l \leq K$. \\
\end{corollary}
%%%%%%%%%%%%%%%%%%%%%%%%%%%%%
\begin{example}
\label{ex5}
Consider the case $K=21~ D=4,~ \lambda=1$ and
$${\cal K}_k = \left\{
                                                \begin{array}{ll}
                                                \{x_{k+4}\},$ if $1\leq k \leq 16\\
                                                \{x_{k+1},x_{k+2},x_{k+3},x_{k+4}\},$ if $17 \leq k \leq 21
                                                \end{array}
                                \right. \\ $$
The proposed code is\\
$\mathfrak{C}=\{x_{1}+x_{5}, x_{2}+x_{6}, x_{3}+x_{7},x_{4}+x_{8}, x_{5}+x_{9}, \\
~~~~~~~~ x_{6}+x_{10}, x_{7}+x_{11}, x_{8}+x_{12}, x_{9}+x_{13}, x_{10}+x_{14}, \\
~~~~~~~~ x_{11}+x_{15}, x_{12}+x_{16}, x_{13}+x_{17}, x_{14}+x_{18}, x_{15}+x_{19}, \\
~~~~~~~~x_{16}+x_{20}, x_{17}+x_{18}+x_{19}+x_{20}+x_{21}\}$. \\
Let \mbox{$m=2$}. Then \mbox{$K=42,~ D=25$} and
$${\cal K}_k = \left\{
                                                \begin{array}{ll}
                                                \{x_{k+21},~ x_{k+4},~ x_{k+25}\},$ if $1\leq k \leq 16,~ 22\leq k \leq 37\\
                                                \{x_{k+1},x_{k+2},x_{k+3},x_{k+4},\\
                                                 x_{k+21},x_{k+22},x_{k+23},x_{k+24},~$ if $17 \leq k \leq 21,
                                                 \\x_{k+25}\}~~~~~~~~~~~~~~~~~~~~~~~~~~~~~~38 \leq k \leq 42
                                                \end{array}
                                \right.\\$$
The code is given by \\
$\mathfrak{C}_{13}$=$\{x_{1}+x_{22}+x_{5}+x_{26},\ ~~x_{2}+x_{23}+x_{6}+x_{27},\\
~~~~~~~~ x_{3}+x_{24}+x_{7}+x_{28},\ ~~x_{4}+x_{25}+x_{8}+x_{29},\\
~~~~~~~~ x_{5}+x_{26}+x_{9}+x_{30},\ ~~x_{6}+x_{27}+x_{10}+x_{31},\\
~~~~~~~~ x_{7}+x_{28}+x_{11}+x_{32},\ ~x_{8}+x_{29}+x_{12}+x_{33},\\
~~~~~~~~ x_{9}+x_{30}+x_{13}+x_{34},\ ~x_{10}+x_{31}+x_{14}+x_{35},\\
~~~~~~~~ x_{11}+x_{32}+x_{15}+x_{36},\ x_{12}+x_{33}+x_{16}+x_{37},\\
~~~~~~~~ x_{13}+x_{34}+x_{17}+x_{38},\ x_{14}+x_{35}+x_{18}+x_{39},\\
~~~~~~~~ x_{15}+x_{36}+x_{19}+x_{40},\ x_{16}+x_{37}+x_{20}+x_{41},\\
~~~~~~~~~~~x_{17}+x_{38}+x_{18}+x_{39}+x_{19}+x_{40}+x_{20}+x_{41}+x_{21}+x_{42}\}.$\\

Let \mbox{$m=3$}. Then \mbox{$K=63,~ D=46$} and
$${\cal K}_k = \left\{
                                                \begin{array}{ll}
                                                \{x_{k+21},x_{k+42},x_{k+4},x_{k+25},x_{k+46}\},$ if $1\leq k \leq 16,\\
                                                ~~~~~~~~~~~22\leq k \leq 37,~ 43\leq k \leq 58\\
                                                \{x_{k+1},x_{k+2},x_{k+3},x_{k+4},\\ x_{k+21},x_{k+22},x_{k+23},x_{k+24},       \\
                                                x_{k+25},x_{k+42},x_{k+43},x_{k+44},~$ ~~~~~~~~if $17 \leq k \leq 21,\\
                                                x_{k+45},x_{k+46}\}~~~~~~~~~~~~~~~~38 \leq k \leq 42,~ 59 \leq k \leq 63
                                                \end{array}
                                \right.\\$$
The code is given by \\
$\mathfrak{C}_{14}$=$\{x_{1}+x_{22}+x_{43}+x_{5}+x_{26}+x_{47},\\
x_{2}+x_{23}+x_{44}+x_{6}+x_{27}+x_{48},\\
x_{3}+x_{24}+x_{45}+x_{7}+x_{28}+x_{49},\\
x_{4}+x_{25}+x_{46}+x_{8}+x_{29}+x_{50},\\
x_{5}+x_{26}+x_{47}+x_{9}+x_{30}+x_{51},\\
x_{6}+x_{27}+x_{48}+x_{10}+x_{31}+x_{52},\\
x_{7}+x_{28}+x_{49}+x_{11}+x_{32}+x_{53},\\
x_{8}+x_{29}+x_{50}+x_{12}+x_{33}+x_{54},\\
x_{9}+x_{30}+x_{51}+x_{13}+x_{34}+x_{55},\\
x_{10}+x_{31}+x_{52}+x_{14}+x_{35}+x_{56},\\
x_{11}+x_{32}+x_{53}+x_{15}+x_{36}+x_{57},\\
x_{12}+x_{33}+x_{54}+x_{16}+x_{37}+x_{58},\\
x_{13}+x_{34}+x_{55}+x_{17}+x_{38}+x_{59},\\
x_{14}+x_{35}+x_{56}+x_{18}+x_{39}+x_{60},\\
x_{15}+x_{36}+x_{57}+x_{19}+x_{40}+x_{61},\\
x_{16}+x_{37}+x_{58}+x_{20}+x_{41}+x_{62},\\
x_{17}+x_{38}+x_{59}+x_{18}+x_{39}+x_{60}+x_{19}+x_{40}+x_{61}+x_{20}+x_{41}+x_{62}+x_{21}+x_{42}+x_{63}\}.$
\end{example}
%%%%%%%%%%%%%%%%%%%%%%%%%%%%%%%%%%%%%%
\noindent
{\bf Class (iii):} For this class of problems there is an integer $\lambda$ such that $D+\lambda$ divides $K$ and $\lambda$ divides $D$ with the antidote being ${\cal K}_k=\{x_{k+\lambda},x_{k+2\lambda},...,x_{k+D}\}.$ An optimal linear index code for this class of problems is known to be \cite{MRRarXiv} \\
$\mathfrak{C}=\{x_{i+j\lambda}+x_{i+(j+1)\lambda}+\dots+x_{i+(j+p)\lambda}\\
~~~~~~~~~~~~~~~  | i =\{1,2,\dots,\lambda\},\ j= \{1,2,\dots,\frac{K-D-\lambda}{\lambda}$\}\} \\
where  $\frac{D}{\lambda}=p$ and $\frac{K}{D+\lambda}=n.$ \\

\begin{corollary}
\label{cor6}
The lifting construction for a problem in this class gives a problem with $mK$ messages having antidotes for $1 \leq k \leq mK$ given by    
$${\cal K}_k=\left\{
                \begin{array}{ll}
                  \ x_{k+K},x_{k+2K},x_{k+3K},\cdots,x_{k+(m-1)K} \\
                  \ x_{k+\lambda},x_{k+\lambda +K},x_{k+\lambda +2K},\cdots,x_{k+\lambda +(m-1)K}\\
                  \ x_{k+2\lambda},x_{k+2\lambda +K},x_{k+2\lambda +2K},\cdots,x_{k+2\lambda+(m-1)K} \\
                  \ \vdots \\
                  \ x_{k+D},x_{k+D+K},x_{k+D+2K},\cdots,x_{k+D+(m-1)K}. \\
                  \end{array}
              \right.
$$
with  scalar linear code for the lifted problem given by \\
$\mathfrak{C}^{(m)}=\{x_{i+j\lambda}^\prime +x_{i+(j+1)\lambda}^\prime +\dots+x_{i+(j+p)\lambda}^\prime  \\ 
~~~~~~~~~~~~~~~~~~~~ | i =\{1,2,\dots,\lambda\},\ j= \{1,2,\dots,\frac{K-D-\lambda}{\lambda}$\}\} \\
where  $\frac{D}{\lambda}=p$, $\frac{K}{D+\lambda}=n$ and \\
$~~~~~~~~~ x_{l}^\prime =x_{l}+x_{l+K}+x_{l+2K}+\dots+x_{l+(m-1)K}$ \\
for $l=1,2,\cdots,K$. \\
\end{corollary}
%%%%%%%%%%%%%%%%%%%%%%%%%%%%%%%%%%%%%%%%%%%
\begin{example}
\label{ex6}
Consider the case $K=18,~ D=5,~ \lambda=1$ and ${\cal K}_k=\{x_{k+1},x_{k+2},\cdots,x_{k+5}\}$ for $k=1,2,\cdots,18.$\\
The optimal linear code for this case is\\
$\mathfrak{C}= \{
x_{1}+x_{2}+x_{3}+x_{4}+x_{5}+x_{6},~ x_{2}+x_{3}+x_{4}+x_{5}+x_{6}+x_{7}, \\
x_{3}+x_{4}+x_{5}+x_{6}+x_{7}+x_{8}, ~ x_{4}+x_{5}+x_{6}+x_{7}+x_{8}+x_{9}, \\
x_{5}+x_{6}+x_{7}+x_{8}+x_{9}+x_{10}, ~ x_{6}+x_{7}+x_{8}+x_{9}+x_{10}+x_{11}, \\
x_{7}+x_{8}+x_{9}+x_{10}+x_{11}+x_{12}, ~ x_{8}+x_{9}+x_{10}+x_{11}+x_{12}+x_{13}, \\
x_{9}+x_{10}+x_{11}+x_{12}+x_{13}+x_{14},\\ x_{10}+x_{11}+x_{12}+x_{13}+x_{14}+x_{15},\\
x_{11}+x_{12}+x_{13}+x_{14}+x_{15}+x_{16},\\ x_{12}+x_{13}+x_{14}+x_{15}+x_{16}+x_{17},\\
x_{13}+x_{14}+x_{15}+x_{16}+x_{17}+x_{18}\}.$\\

\noindent
Let~\mbox{$m=2$}. Then \mbox{$K=36,~ D=23$} and for $k=1,2,\cdots,36,$ \\
\mbox{${\cal K}_k=\{x_{k+1},x_{k+2},\cdots,x_{k+5},x_{k+19},x_{k+20},\cdots,x_{k+23}\}$} \\
The lifted code is given by \\
\begin{small}
{
$\mathfrak{C^{(2)}}= \{
x_{1}+x_{19}+x_{2}+x_{20}+x_{3}+x_{21}+x_{4}+x_{22}+x_{5}+x_{23}+x_{6}+x_{24},\\
x_{2}+x_{20}+x_{3}+x_{21}+x_{4}+x_{22}+x_{5}+x_{23}+x_{6}+x_{24}+x_{7}+x_{25},\\
x_{3}+x_{21}+x_{4}+x_{22}+x_{5}+x_{23}+x_{6}+x_{24}+x_{7}+x_{25}+x_{8}+x_{26},\\
x_{4}+x_{22}+x_{5}+x_{23}+x_{6}+x_{24}+x_{7}+x_{25}+x_{8}+x_{26}+x_{9}+x_{27},\\
x_{5}+x_{23}+x_{6}+x_{24}+x_{7}+x_{25}+x_{8}+x_{26}+x_{9}+x_{27}+x_{10}+x_{28},\\
x_{6}+x_{24}+x_{7}+x_{25}+x_{8}+x_{26}+x_{9}+x_{27}+x_{10}+x_{28}+x_{11}+x_{29},\\
x_{7}+x_{25}+x_{8}+x_{26}+x_{9}+x_{27}+x_{10}+x_{28}+x_{11}+x_{29}+x_{12}+x_{30},\\
x_{8}+x_{26}+x_{9}+x_{27}+x_{10}+x_{28}+x_{11}+x_{29}+x_{12}+x_{30}+x_{13}+x_{31},\\
x_{9}+x_{27}+x_{10}+x_{28}+x_{11}+x_{29}+x_{12}+x_{30}+x_{13}+x_{31}+x_{14}+x_{32},\\
x_{10}+x_{28}+x_{11}+x_{29}+x_{12}+x_{30}+x_{13}+x_{31}+x_{14}+x_{32}+x_{15}+x_{33},\\
x_{11}+x_{29}+x_{12}+x_{30}+x_{13}+x_{31}+x_{14}+x_{32}+x_{15}+x_{33}+x_{16}+x_{34},\\
x_{12}+x_{30}+x_{13}+x_{31}+x_{14}+x_{32}+x_{15}+x_{33}+x_{16}+x_{34}+x_{17}+x_{35},\\
x_{13}+x_{31}+x_{14}+x_{32}+x_{15}+x_{33}+x_{16}+x_{34}+x_{17}+x_{35}+x_{18}+x_{36}\}.$\\
}
\end{small}
\end{example}
%%%%%%%%%%%%%%%%%%%%%%%%%%%%%%%%%%%%%%%%%%%%%%%%%%%%%%%%%%%%%%%%%%%%%%
\noindent
{\bf Class (iv):} This class of problems consists of the cases for which there is an integer $\lambda$ such that $D$ divides $K+\lambda$ and $\lambda$ divides $D$ and the antidote is given by
$${\cal K}_k=\left\{
                \begin{array}{ll}
                  \{x_{k+D}\}$, if $\ k\leq K-2D+\lambda\\
                  \{x_{k+\lambda}, x_{k+2\lambda},...,x_{k+D}\}$, if $K-2D+\lambda<k\leq K,
                  \end{array}
              \right.$$
In \cite{MRRarXiv} it shown that an optimal index code for this class is \\
 $\mathfrak{C}=\{{x_{i+(j-1)D}+x_{i+jD}}|i = 1,2,\dots,D, ~~j = 1,2,\dots,n-2\}\\ 
~~~~~~~~~ \cup \{x_{K-2D+1+\lambda+i^\prime }+x_{K-D+1+i^\prime }+x_{K-\lambda+1+i^\prime  mod \lambda}| \\
~~~~~~~~~~~~~~~~~~~~~~~~~~~~~~~~~~ i^\prime  = \{0,1,2,\dots,p-1\}\}$\\
where $\frac{K+\lambda}{D}=n(>2)$, $p=K$ mod $D$ = $D-\lambda$.

\begin{corollary}
\label{cor7}
For an index coding problem of this class the lifting construction with parameter $m$ gives the index coding problem with $mK$ number of messages and the antidote pattern for $1 \leq k \leq mK$ given as follows:\\
If $(k \ mod \ K)\leq K-2D+\lambda,$ then 
$${\cal K}_k=\left\{
                \begin{array}{ll}
                  \ x_{k+K},x_{k+2K},x_{k+3K},\cdots,x_{k+(m-1)K} \\
				  \ x_{k+K+D},x_{k+2K+D},x_{k+3K+D},\cdots,x_{k+(m-1)K+D} \\ 
                  \end{array}
              \right.
$$
and if $K-2D+\lambda<(k \ mod \ K)\leq K,$ then
$${\cal K}_k=\left\{
                \begin{array}{ll}
                  \ x_{k+K},x_{k+2K},x_{k+3K},\cdots,x_{k+(m-1)K} \\
                  \ x_{k+\lambda},x_{k+\lambda +K},x_{k+\lambda +2K},\cdots,x_{k+\lambda +(m-1)K}\\
                  \ x_{k+2\lambda},x_{k+2\lambda +K},x_{k+2\lambda +2K},\cdots,x_{k+2\lambda+(m-1)K} \\
                  \ \vdots \\
                  \ x_{k+D},x_{k+D+K},x_{k+D+2K},\cdots,x_{k+D+(m-1)K}. \\
                  \end{array}
              \right.
$$

An optimal  lifted  scalar linear code is \\
$\mathfrak{C}^{(m)}=\{{x_{i+(j-1)D}^\prime +x_{i+jD}^\prime }| 1 \leq i \leq D, ~~ 1 \leq j \leq n-2 \}\\
~~~  \cup \{x_{K-2D+1+\lambda+i^\prime }^\prime +x_{K-D+1+i^\prime }^\prime +x_{K-\lambda+1+i^\prime  mod \lambda}^\prime  \\
~~~~~~~~~~~~~~~~~~ | 0 \leq i^\prime  \leq p-1\}$ \\
where $\frac{K+\lambda}{D}=n(>2)$, $p=K$ mod $D$ = $D-\lambda$ and \\
$ ~~~~~~x_{l}^\prime =x_{l}+x_{l+K}+x_{l+2K}+\dots+x_{l+(m-1)K}$ \\ 
for $l=1,2,\cdots,K$.
\end{corollary}
%%%%%%%%%%%%%%%%%%%%%%%%%%%%%%%%%%%%%%%
\begin{example}
\label{ex7}
Consider the case $K=19,~ D=5,~ \lambda=1,$ and 
$${\cal K}_k = \left\{
						\begin{array}{ll}
						\{x_{k+5}\},$ if $1\leq k \leq 10\\
						\{x_{k+1},x_{k+2},\dots,x_{k+5}\},$ if $11 \leq k \leq 19.
						\end{array}
              			\right. \\ $$
The index code for this problem is \\
$\mathfrak{C}=\{x_{1}+x_{6},\ x_{6}+x_{11},\ x_{11}+x_{15}+x_{19},\\ x_{2}+x_{7},\ x_{7}+x_{12},\ x_{12}+x_{16}+x_{19},\\ x_{3}+x_{8},\ x_{8}+x_{13},\ x_{13}+x_{17}+x_{19},\\ x_{4}+x_{9},\ x_{9}+x_{14},\ x_{14}+x_{18}+x_{19}\}.\\ x_{5}+x_{10},\ x_{10}+x_{15},$\\

\noindent
Let \mbox{$m=2$}. Then the lifting construction leads to the problem with \mbox{$K=38,~ D=24$} and 
$${\cal K}_k = \left\{
						\begin{array}{ll}
						\{x_{k+19},~ x_{k+5},~ x_{k+24}\},$ if $1\leq k \leq 10,\\
						~~~~~~~~~~~~~~~~~~~~~~~~~~~~~20 \leq k \leq 29\\
						\{x_{k+1},x_{k+2},\dots,x_{k+5},\\
						 x_{k+19},x_{k+20},x_{k+21},\dots,x_{k+24}\},~$if $11 \leq k \leq 19,\\
						 ~~~~~~~~~~~~~~~~~~~~~~~~~~~~~~~~~~~~~~~~30 \leq k \leq 38
						\end{array}
              			\right.\\$$
The lifted code is given by \\
$\mathfrak{C}^{(2)}=\{x_{1}+x_{20}+x_{6}+x_{25},~x_{2}+x_{21}+x_{7}+x_{26},\\
x_{3}+x_{22}+x_{8}+x_{27},~x_{4}+x_{23}+x_{9}+x_{28},\\
x_{5}+x_{24}+x_{10}+x_{29},~x_{6}+x_{25}+x_{11}+x_{30},\\
x_{7}+x_{26}+x_{12}+x_{31},~x_{8}+x_{27}+x_{13}+x_{32},\\
x_{9}+x_{28}+x_{14}+x_{33},~x_{10}+x_{29}+x_{15}+x_{34},\\
x_{11}+x_{30}+x_{15}+x_{34}+x_{19}+x_{38},\\
x_{12}+x_{31}+x_{16}+x_{35}+x_{19}+x_{38},\\
x_{13}+x_{32}+x_{17}+x_{36}+x_{19}+x_{38},\\
x_{14}+x_{33}+x_{18}+x_{37}+x_{19}+x_{38}\}.$\\

\noindent
Let \mbox{$m=3$}. Then have the parameters \mbox{$K=57,~ D=43$} for the lifted problem with the antidotes 
$${\cal K}_k = \left\{
						\begin{array}{ll}
						\{x_{k+19},x_{k+38},x_{k+5},x_{k+24},x_{k+43}\},$ if $1\leq k \leq 10,\\~~~~~~~~~~~~~~~~~~~~~~~~~~~~20\leq k \leq 29,~ 39 \leq k \leq 48\\
						\{x_{k+1},x_{k+2},\dots,x_{k+5}\\
						 x_{k+19},x_{k+20},x_{k+21},\dots,x_{k+24}, $ if $11 \leq k \leq 19,\\
						 x_{k+38},x_{k+39},x_{k+40},\dots,x_{k+43}\},\ 30 \leq k \leq 38,\ 49 \leq k \leq 57
						\end{array}
              			\right.\\$$
The lifted code is given by \\
$\mathfrak{C}^{(3)}=\{x_{1}+x_{20}+x_{39}+x_{6}+x_{25}+x_{44},\\
x_{2}+x_{21}+x_{40}+x_{7}+x_{26}+x_{45},\\
x_{3}+x_{22}+x_{41}+x_{8}+x_{27}+x_{46},\\
x_{4}+x_{23}+x_{42}+x_{9}+x_{28}+x_{47},\\
x_{5}+x_{24}+x_{43}+x_{10}+x_{29}+x_{48},\\
x_{6}+x_{25}+x_{44}+x_{11}+x_{30}+x_{49},\\
x_{7}+x_{26}+x_{45}+x_{12}+x_{31}+x_{50},\\
x_{8}+x_{27}+x_{46}+x_{13}+x_{32}+x_{51},\\
x_{9}+x_{28}+x_{47}+x_{14}+x_{33}+x_{52},\\
x_{10}+x_{29}+x_{48}+x_{15}+x_{34}+x_{53},\\
x_{11}+x_{30}+x_{49}+x_{15}+x_{34}+x_{53}+x_{19}+x_{38}+x_{57},\\
x_{12}+x_{31}+x_{50}+x_{16}+x_{35}+x_{54}+x_{19}+x_{38}+x_{57},\\
x_{13}+x_{32}+x_{51}+x_{17}+x_{36}+x_{55}+x_{19}+x_{38}+x_{57},\\
x_{14}+x_{33}+x_{52}+x_{18}+x_{37}+x_{56}+x_{19}+x_{38}+x_{57}.\}$\\
\end{example}
%%%%%%%%%%%%%%%%%%%%%%%%%%%%%%%%%%%%%
\section{Discussion}
\label{sec6}
In this paper a lifting construction is given for scalar linear index codes of multiple unicast index problems which results in a sequence of index coding problems with integer multiple number of messages and receivers  for any arbitrary integer. Moreover, it is shown that if the problem with which the lifting begins  has an optimal linear index code then it induces an optimal linear index code for the lifted problem. This lifting construction has been used on ten classes of index coding problems for which optimal linear index codes are known and new classes codes have been obtained starting from these classes of codes.

The side information in a multicast index coding problem is represented by a directed graph G = ($V$,$E$) with $V = \{1,2,...,K\}$ is the set of vertices and E is the set of edges such that the directed edge $(i,j)\in E$ if receiver (destination) $R_{i}$ knows $x_{j}$. This graph $G$ for a given index coding problem is called side information graph \cite{YBJK}. In a given index coding problem with side information graph $G$ = $(V,E)$, an edge $e\in E$ is said to be critical if the removal of $e$ from $G$ strictly reduce the capacity. The index coding problem $G=(V,E)$ is critical if every $e\in E$ is critical \cite{TSG}. An interesting problem to pursue is to find classes of codes which under lifting lead to critical index coding problems. 

Another interesting direction of further research is to study the suitability of the new classes of codes presented in this paper for application to noise broadcasting problem. Recently, it has been observed that in a noisy index coding problem it is desirable for the purpose of reducing the probability of error that  the receivers use as small a number of transmissions from the source as possible and linear index codes with this property have been reported in \cite{TRCR}, \cite{KaR}. While the report \cite{TRCR} considers fading broadcast channels, in \cite{AnR} AWGN channels are considered and it is reported that linear index codes with minimum length (capacity achieving codes or optimal length codes) help to facilitate to achieve more reduction in probability of error compared to non-minimum length codes for receivers with large amount of side-information. These aspects remain to be investigated for the new classes of sequences of codes presented in this paper.
%%%%%%%%%%%%%%%%%%%%%%%%%%%%%%%%%%%%%%%%%%%%%%%

%%%%%%%%%%%%%%%%%%%%%%%%%%%%%%%%

%%%%%%%%%%%%%%%%%%%%%%%%%%%%%%%%%%%%%%%%%%%%%%%
\end{document}